%% file: main.tex
\documentclass[11pt]{article}

\input{macros}
\title{When Relaxation Doesn’t Help: RLDCs with Small Soundness Error Yield LDCs}
\author{
Kuan Cheng\thanks{\texttt{ckkcdh@pku.edu.cn}, School of Computer Science, Peking University. 
}
\and
Xin Li\thanks{\texttt{lixints@cs.jhu.edu}, Department of Computer Science, Johns Hopkins University. Supported in part by NSF Award CCF-2127575.}
\and
Songtao Mao\thanks{\texttt{smao13@jhu.edu}, Department of Computer Science, Johns Hopkins University. Supported in part by NSF Award CCF-2127575.}
}
\date{}

\begin{document}

\maketitle

\input{abstract}
\newpage
\tableofcontents
\newpage
\input{Introduction}
\input{Preliminary}
\input{Results}
\input{LBPCPPRLDC}

\bibliographystyle{alpha}
\bibliography{references}

\end{document}

%% file: macros.tex
\usepackage[margin=1in]{geometry}
\usepackage{xcolor}
\usepackage{amsfonts}
\usepackage{amssymb}
\usepackage{amsmath}
\usepackage{array}
\usepackage{verbatim}
\usepackage{graphicx}
\usepackage{mathrsfs} 
\usepackage{algpseudocode}
\usepackage{tikz}
\usetikzlibrary{decorations.pathreplacing}

\usepackage{hyperref} 
\usepackage[ruled,vlined]{algorithm2e}
\usepackage{enumitem}

\usepackage{amsthm}

\newtheorem{theorem}{Theorem}[section]
\newtheorem{corollary}[theorem]{Corollary}
\newtheorem{lemma}[theorem]{Lemma}
\newtheorem{proposition}[theorem]{Proposition}

\theoremstyle{definition}
\newtheorem{definition}[theorem]{Definition}

\newtheorem{example}[theorem]{Example}

\theoremstyle{remark}
\newtheorem{remark}[theorem]{Remark}

\newcommand{\sQ}{\mathcal{Q}}

\newcommand{\NN}{\mathbb{N}}

\DeclareSymbolFont{bbold}{U}{bbold}{m}{n}
\DeclareSymbolFontAlphabet{\mathbbold}{bbold}

\DeclareMathOperator*{\E}{\mathbb{E}}

\newcommand{\supp}{\mathrm{supp}}

\newcommand{\NS}{\mathsf{NS}}
\newcommand{\Err}{\mathrm{Err}}
\newcommand{\Dec}{\mathsf{Dec}}

\title{RLDC is equivalent to LDC in constant query settings}

\date{}

%% file: abstract.tex
\begin{abstract}
Locally decodable codes (LDCs) are error-correcting codes that allow recovery of any single message symbol by probing only a small number of positions from the (possibly corrupted) codeword. Relaxed locally decodable codes (RLDCs) further allow the decoder to output a special failure symbol $\bot$ on a corrupted codeword. While known constructions of RLDCs achieve much better parameters than standard LDCs, it is intriguing to understand the relationship between LDCs and RLDCs. On the one hand, separation results (i.e., the existence of $q$-query RLDCs that are not $q$-query LDCs) are known for $q=3$ \cite{GMWZ25} and $q \geq 15$ \cite{grigorescu2025relaxed}. On the other hand, prior work \cite{block2022relaxed} shows that any $2$-query RLDC also gives a $2$-query LDC, and \cite{grigorescu2025relaxed} shows that any \emph{linear} $3$-query RLDC is also a linear $3$-query LDC. Furthermore, \cite{grigorescu2025relaxed} shows that when the soundness error of a \emph{linear} $q$-query RLDC with perfect completeness is below some threshold $s(q)$, the code must also be a linear $q$-query LDC with comparable parameters.

In this work, we extend the main result in \cite{grigorescu2025relaxed} by removing the linearity requirement in the nonadaptive setting. Specifically, we show that every nonadaptive $(q,\delta,1,s)$-RLDC over a finite alphabet $\Sigma$ with $s<|\Sigma|^{-q}$ yields a $q$-query LDC with comparable decoding radius and error. Our results also extend to the setting of locally correctable codes (LCCs) and relaxed locally correctable codes (RLCCs). From this, we also obtain lower bounds for nonadaptive RLDCs from known LDC lower bounds.
\end{abstract}

%% file: Introduction.tex
\section{Introduction}
Error correction codes are fundamental objects in computer science that enable the correct transmission of messages over a noisy or adversarial channel. While standard decoding algorithms need to look at the entire received codeword, locally decodable codes (LDCs) allow super-fast recovery of individual symbols from a corrupted codeword. In particular, given an LDC \( C: \Sigma^k \rightarrow \Sigma^n \), any symbol of the original message \( x \in \Sigma^k \) can be recovered by querying only a few positions of the corrupted codeword \( y \in \Sigma^n \), even in the presence of worst-case errors. For clarity we focus on the binary alphabet in the introduction; the RLDC-to-LDC and RLCC-to-LCC results extend to arbitrary finite alphabets. We have the following definition.

\begin{definition}[Locally decodable codes (binary LDCs)]
A code $C : \{0,1\}^k \to \{0,1\}^n$ is a
$(q,\delta,s)$-\emph{locally decodable code} if there exists a
$q$-query local decoder
\[
    \Dec : \{0,1\}^n \times [k] \to \{0,1\}
\]
such that for every message $b \in \{0,1\}^k$, every index $i \in [k]$,
and every received word $y \in \{0,1\}^n$ with
$\Delta(y,C(b)) \le \delta n$, we have
\[
    \Pr\bigl[\Dec^y(i) \neq b_i\bigr] \;\le\; s.
\]
\end{definition}

Similarly, a locally correctable code (LCC) allows the recovery of any symbol of the codeword \( C(x) \) by querying only a few positions of the corrupted codeword. Since their introduction by Katz and Trevisan \cite{katz2000efficiency}, these codes have found applications in many areas of computer science, such as private information retrieval, probabilistically checkable proofs, fault-tolerant circuits,  hardness amplification, and data structures (see e.g., \cite{BabaiFLS91,LundFKN92,BlumLR93,BlumK95,ChorKGS98,ChenGW13,AndoniLRW17, Tre04-survey,Gasarch04}).

In the study of LDC/LCCs, one generally assumes that the code can tolerate some constant fraction of Hamming errors. Then, the most important parameters of these codes are the \emph{rate} of the code, defined as the ratio between the message length $k$ and the codeword length $n$, and the number of queries that the decoding algorithm needs to make. On the other hand, with only a few queries allowed, the decoding algorithm is necessarily randomized, hence there is also a small probability (say $< 1/3$) of decoding error when the codeword is corrupted. 

Despite decades of extensive study on the trade-offs between these parameters, in general there is still a large gap between the upper and lower bounds we have. Specifically, for $2$-query LDCs/LCCs we do have the tight bound of $n=2^{\Theta(k)}$ \cite{kerenidis2003exponential, GoldreichKST06, Ben-AroyaRW08, bhattacharyya2017lower}. However, for $q>2$ queries, the current gap between upper and lower bounds on $n$ as a function of $k$ is superpolynomial in $k$. In particular, the best explicit constructions of constant $q$-query LDCs have super-polynomial but sub-exponential codeword length \cite{yekhanin2008towards,dvir2011matching,Efremenko12}, while the best known constructions of constant $q$-query LCCs have codeword length $n=2^{O(k^{\frac{1}{q-1}})}$, given by the Reed-Muller codes. On the other hand, the most general lower bounds for $q\geq 3$ LDCs/LCCs are of the form $n=\tilde{\Omega}( k^{1+\frac{2}{q-2}})$ \cite{katz2000efficiency,woodruff2007new, Basuldc}. For binary linear $3$-query LCCs, the stronger bound $k=O_\delta(\log^2 n\log\log n)$ is known \cite{alrabiah2024near}.

Given this large gap in our understanding of LDC/LCCs, Ben-Sasson, Goldreich, Harsha, Sudan, and Vadhan \cite{ben2004robust} introduced a relaxed version of LDCs (RLDCs), whose decoder may output a special symbol $\bot$ to indicate a decoding failure. Specifically, we have the following definition.

\begin{definition}[Relaxed LDCs (binary RLDCs)]
A code $C : \{0,1\}^k \to \{0,1\}^n$ is a
$(q,\delta,c,s)$-\emph{relaxed locally decodable code} if there exists a
$q$-query relaxed local decoder
\[
    \Dec : \{0,1\}^n \times [k] \to \{0,1\} \cup \{\bot\}
\]
such that for every $b \in \{0,1\}^k$, every $i \in [k]$, and every
$y \in \{0,1\}^n$ with $\Delta(y,C(b)) \le \delta n$:
\begin{enumerate}
    \item \textbf{(Relaxed completeness)} On the uncorrupted codeword,
    \[
        \Pr\bigl[\Dec^{C(b)}(i) = b_i\bigr] \;\ge\; c.
    \]
    \item \textbf{(Relaxed soundness)} On corrupted words,
    \[
        \Pr\bigl[\Dec^y(i) \notin \{b_i,\bot\}\bigr] \;\le\; s.
    \]
\end{enumerate}
When $c = 1$ we say the relaxed decoder has perfect completeness.
\end{definition}

Remarkably, \cite{ben2004robust} and later Asadi and Shinkar \cite{AsadiS21} construct RLDCs with a constant number $q$ of queries and $n=k^{1+O(1/q)}$, which stands in sharp contrast with the best known constructions of standard LDCs. General RLDC lower bounds appear in \cite{gur2021power,de2021structural}. For linear binary RLDCs, Goldberg, Gur, and Saraogi \cite{GGS25} prove the nearly matching bound $n=k^{1+\Omega(1/q)}$.

This motivates the intriguing question of the relation between LDCs and RLDCs. For example, in \cite{ben2004robust}, the authors asked the question whether one can obtain $q$-query RLDCs with codeword length $n = o(k^{1+1/(q-1)})$, which will imply that the relaxation in RLDCs is necessary, because such codes will beat the known lower-bound for LDCs. Alternatively, it may be possible to improve the lower-bound for ($q$-query) LDCs to provide such a separation as well. A recent work \cite{GMWZ25} indeed achieved this in the case of $q=3$, by constructing explicit $3$-query RLDCs with $n=\tilde{O}(k^2)$. However, no such separation based on improved upper/lower bounds is known for $q>3$ to the best of our knowledge. In fact, \cite{kumar2024relaxed} observes that, in the constant-query regime, the best known RLDC block length matches, up to a constant factor in $q$, the block length lower bound for LDCs. We also mention that the recent work \cite{grigorescu2025relaxed} constructed a $15$-query linear RLDC, and showed that it is not a $q$-query LDC for any constant $q$, although the proof is via a direct combinatorial argument instead of proving improved bounds.

Several recent works \cite{block2022relaxed, grigorescu2025relaxed} approached this question from another direction. In particular, \cite{block2022relaxed} shows that any $2$-query RLDC gives a $2$-query LDC with similar parameters, and hence proves an exponential lower bound on the codeword length of $2$-query RLDCs. \cite{grigorescu2025relaxed} further shows that any $3$-query \emph{linear} RLDC is also a $3$-query \emph{linear} LDC with similar parameters. More generally, any $q$-query linear RLDC with perfect completeness and decoding error $< 2^{- \lfloor q/2 \rfloor}$ is also a $q$-query linear LDC with comparable parameters. This implies that the only way for a perfectly complete $q$-query linear RLDC to have the strong soundness property (i.e., the decoding error is proportional to the fraction of errors, a property that is satisfied by constant-query linear LDCs) is for it to be an LDC already.

In this paper, we generalize the last result in \cite{grigorescu2025relaxed} to nonlinear codes in the nonadaptive, perfectly complete setting. We show that any nonadaptive $q$-query RLDC with perfect completeness and soundness error $< 2^{-q}$ is also a $q$-query LDC with comparable parameters. Thus, in this small-soundness-error regime, relaxation does not buy additional power for nonadaptive decoders.

\subsection{Our results}

For clarity of exposition we first state our RLDC-to-LDC result in the binary setting $\Sigma=\{0,1\}$. The RLDC-to-LDC and RLCC-to-LCC results extend in a straightforward way to arbitrary finite alphabets; the general formulations and proofs appear later in Section~\ref{sec:pre} and the subsequent technical sections.

We begin by recalling the prior implication from relaxed to full local decoding in the low-soundness regime. In \cite{grigorescu2025relaxed}, the authors showed that, for linear RLDCs with perfect completeness, sufficiently small relaxed soundness already forces standard local decodability. More precisely:
\begin{theorem}[{\cite[Theorem 2]{grigorescu2025relaxed}}]
Let $C:\{0,1\}^k\to\{0,1\}^n$ be a linear $(q,\delta,1,s)$-RLDC
with a (possibly adaptive) decoder and $s<2^{-\lfloor q/2\rfloor}$, then, for every radius $r\in \bigl(0,\frac{\delta}{2q}(1-s2^{\lfloor q/2\rfloor}) \bigr)$, $C$ is a 
$$(q,r,\frac{2^{-\lfloor q/2\rfloor}rq}{(2^{-\lfloor q/2\rfloor}-s) \delta})-\text{LDC}.$$
Moreover, the analogous implication holds in the RLCC/LCC setting.
\end{theorem}

Thus, any linear RLDC with perfect completeness and error below the threshold $2^{-\lfloor q/2\rfloor}$ is automatically an LDC with comparable parameters.
Our main result removes the linearity requirement at the cost of assuming nonadaptivity and using the threshold $2^{-q}$. Roughly speaking, whenever the code has perfect completeness, a nonadaptive relaxed decoder, and soundness error below $2^{-q}$, it already admits a standard local decoder with the same query complexity.

\begin{theorem}[Informal, binary case of Theorem \ref{thm:derived_ldc}]
\label{thm:binary-main}
Let $C : \{0,1\}^k \to \{0,1\}^n$ be a $(q,\delta,1,s)$-RLDC with a
nonadaptive relaxed decoder, and assume $s < 2^{-q}$.
Then for every error radius $r \in  \bigl(0,\frac{\delta}{2q}(1-s2^q) \bigr)$ there exists a
$q$-query local decoder $\Dec'$ such that $C$ is a
\[
    \bigl(q, r,\;  \tfrac{s 2^{q}}{2} + \tfrac{r q}{\delta}\bigr)\text{-LDC}
\]
with respect to $\Dec'$.
\end{theorem}

Some exponential dependence of the soundness threshold on $q$ is necessary. Indeed, \cite{grigorescu2025relaxed} exhibits a linear RLDC with constant query complexity that is not an LDC. Independent repetition decreases its soundness error exponentially while increasing its query complexity only linearly, so a threshold with weaker dependence on the query count would eventually turn this code into an LDC, contradicting the construction.

Combining Theorem~\ref{thm:binary-main} with recent lower bounds for LDCs gives lower bounds for nonadaptive RLDCs in the same soundness-error regime.
\begin{corollary} For every constant odd integer $q \ge 3$ and constant $\delta \in (0, 1)$,
 if $C: \{0, 1\}^k \rightarrow \{0, 1\}^n$ is a nonadaptive $(q, \delta, 1, s)$-RLDC with $s\le 2^{-(q+1)}$, then 
$k \le O( n^{1-\frac{2}{q}} \log n  )$.
\end{corollary}

\subsection{Techniques overview}
In this section we give a high-level overview of our proofs.
For concreteness we focus on the binary setting $\Sigma = \{0,1\}$ and on
decoding message symbols.
We first recall the main ideas from \cite{grigorescu2025relaxed} and then describe our improvement.

\paragraph{Summary of techniques in \cite{grigorescu2025relaxed}.}
The core idea is to transform an RLDC decoder into a \emph{smooth} (nonadaptive) local decoder. Concretely, they construct a decoder $\Dec_{\mathsf{LDC}}$ such that for every $i\in[k]$ and $j\in[n]$,
\[
\Pr\!\big[\Dec_{\mathsf{LDC}}(i)\ \text{queries}\ j\big]\le \frac{1}{\eta n},
\]
for some constant $\eta>0$.

A standard observation due to \cite{katz2000efficiency} is that if a smooth nonadaptive decoder additionally has \emph{perfect completeness}, then it already yields an LDC for the code with comparable parameters.
Thus, their main technical contribution is to start from the RLDC decoding procedure and modify it into a smooth nonadaptive decoder that achieves perfect completeness.

A nonadaptive RLDC decoder for bit $b_i$ of a possibly nonlinear code is specified by a distribution $\mathcal{Q}_i$ over subsets $Q \subseteq [n]$ of size at most $q$ and, for each $Q$ in the support of $\mathcal{Q}_i$, a local rule (truth table) $f_{i,Q} : \{0,1\}^Q \to \{0,1,\bot\}$. On input a received word $y \in \{0,1\}^n$, the decoder samples $Q \sim \mathcal{Q}_i$ and outputs $f_{i,Q}(y|_Q)$.

In the special case where $C$ is \emph{linear} and the decoder has \emph{perfect completeness},
the local rule $f_{i,Q}$ can be viewed as having the following structure.
Associated with each query set $Q$ is a collection of \emph{linear consistency checks} on the
symbols in $Q$ (equivalently, a system of linear constraints that valid restrictions
$(C(b))|_Q$ must satisfy). Given $y|_Q$, the decoder first verifies these linear constraints;
if any check fails, it outputs $\bot$. Otherwise, the
desired message bit $b_i$ is determined as a \emph{linear function} of the queried symbols,
so the decoder outputs that linear combination of $y|_Q$.

For each $i\in[k]$, the codeword coordinates are divided into \emph{heavy} and \emph{light} positions.
For each $j \in [n]$ let
$p_i(j) := \Pr_{Q \sim \mathcal{Q}_i}[j \in Q]$ be the probability that coordinate $j$ is queried
when decoding $b_i$.
Positions with
$p_i(j) > q/(\delta n)$ are declared heavy, the rest are light.
By averaging, there can be at most $\delta n$ heavy positions, so changing all heavy coordinates
introduces at most $\delta n$ Hamming errors.

Based on the decoding constraints, they define \emph{$i$-smoothable} query sets.
Intuitively, $Q$ is called $i$-smoothable if $b_i$ can be recovered using only light positions:
there exists a subset $T \subseteq Q \cap L_i$ such that
\[
    b_i \;=\; \bigoplus_{j \in T} c_j
\]
for every codeword $c = C(b)$. If no such representation exists, $Q$ is called $i$-nonsmoothable.

Then they define
\[
p_{i,\mathrm{good}}
\;=\;
\Pr_{Q\sim \sQ_i}\bigl[\,Q \text{ is smoothable}\,\bigr],
\]
the probability that a random query set is smoothable. Under the assumption $s\le(1-\alpha)2^{-\lfloor q/2\rfloor}$ for a constant $\alpha>0$, their argument shows that $p_{i,\mathrm{good}}=\Omega(1)$. They choose a random codeword in an appropriate subspace and copy its heavy coordinates, changing at most $\delta n$ positions. For every nonsmoothable $Q$, this preserves the light view but makes the local constraint agree with a codeword having the opposite value of $b_i$ with probability bounded below in terms of $q$. Averaging over the corruption and $Q$, the RLDC soundness guarantee bounds $1-p_{i,\mathrm{good}}$.

Finally, by restricting to smoothable queries from $\sQ$, the modified decoder $\Dec_{\mathrm{LDC}}$ remains smooth, and thus yields an LDC.

\paragraph{Smoothability relative to a codeword}
Smoothness of a query distribution remains well defined without linearity, but smoothability can depend on the codeword. As Example~\ref{ex:truthtable} shows, a query set can be smoothable for one codeword but not for another. Consequently, we cannot restrict the query distribution to a single codeword-independent collection of smoothable sets. The decoder we derive is therefore not smooth: we keep the original query distribution.

This is our starting point. We formulate a \emph{codeword-relative} notion of smoothability. Specifically, for a codeword $c = C(b) \in \{0,1\}^n$, we say that a query set $Q$ is \emph{$(i,c)$-smoothable} if the restriction of $c$ to the light coordinates $L(Q)$ uniquely determines the message symbol $b_i$ among all codewords; namely,
\[
    \forall b' \in \{0,1\}^k,\;
    c' := C(b') \text{ with } c'|_{L(Q)} = c|_{L(Q)}
    \;\Longrightarrow\;
    b'_i = b_i.
\]
Equivalently, $Q$ is \emph{$(i,c)$-nonsmoothable} if there exists $b' \in \{0,1\}^k$ such that, letting $c' := C(b')$, we have $c'|_{L(Q)} = c|_{L(Q)}$ but $b'_i \neq b_i$. In other words, the light restriction $c|_{L(Q)}$ admits a completion to a different codeword whose $i$-th message symbol disagrees with $b_i$.

In this setting, two natural questions arise:
\begin{enumerate}
    \item \textbf{Local recognizability:} Can we determine whether $Q$ is $(i,c)$-smoothable using only local information, by reading $c$ on the coordinates in $Q$, rather than requiring global knowledge of the codeword?
    \item \textbf{Local dependence of decoding:} Can the decoder's output on a query set $Q$ be made to depend only on the local view $c|_{Q}$ (and perhaps its internal randomness), rather than on any global properties of the underlying codeword $c$?
\end{enumerate}

For the first question: Yes, provided that the light coordinates are uncorrupted.
By definition, whether a query set $Q$ is $(i,c)$-smoothable depends only on the
restriction $c|_{L(Q)}$. Hence, if the received word $y$ satisfies
$y|_{L(Q)} = c|_{L(Q)}$, then the decoder can correctly determine the
$(i,c)$-smoothability status of $Q$ using only local information.

It remains to argue that light corruption is unlikely when the overall error rate is small. Suppose $y$ differs from $c$ in at most $rn$ positions. For a fixed query distribution over $Q$, the probability that $L(Q)$ intersects the corruption set can be bounded by a standard double-counting argument combined with Markov's inequality:
\[
\Pr\bigl[\,y|_{L(Q)} \neq c|_{L(Q)}\,\bigr]
=
\Pr\bigl[\,L(Q)\cap \mathrm{Err}\neq \emptyset\,\bigr]
\;\le\;
\E\bigl[|L(Q)\cap \mathrm{Err}|\bigr]
\;\le\;
 |\Err| \cdot \frac{q}{\delta n}
\;\le\;
\frac{rq}{\delta},
\]
where $\mathrm{Err} := \{j\in[n] : y_j\neq c_j\}$.
Thus, when the error fraction is sufficiently small, this ``bad'' event has small probability. In the LDC analysis we absorb this probability into the overall decoding error.

For the second question: Yes. Our decoder proceeds under the event that no light error occurs, i.e., $y|_{L(Q)}=c|_{L(Q)}$, and then determines whether $Q$ is smoothable using the criterion above. If $Q$ is $(i,c)$-smoothable, the decoder outputs $b_i$ as a deterministic function of the local view on $Q$, using the corresponding truth-table rule together with the queried values $y|_{Q}$.

If $Q$ is $(i,c)$-nonsmoothable (which may occur in our setting, as argued earlier), our decoder instead outputs a uniformly random symbol from the alphabet. This differs from the approach of~\cite{grigorescu2025relaxed}, where the decoder is designed to never query such nonsmoothable sets in the first place.

In summary, for each message index $i$ and query set $Q$, we derive a randomized local rule $g_{i,Q}$ with inputs in $\{0,1\}^{Q}$ and outputs in $\{0,1\}$ from the original RLDC rule $f_{i,Q}:\{0,1\}^{Q}\to \{0,1\}\cup\{\bot\}$ by projecting through the light coordinates.

Fix $i,Q$, and write $L(Q)\subseteq Q$ for the light positions. For any light
pattern $a\in \{0,1\}^{L(Q)}$, consider all completions $z\in \{0,1\}^{Q}$ with
$z|_{L(Q)}=a$. We declare $a$ to be \emph{good} if, among those completions,
whenever $f_{i,Q}$ outputs a symbol (i.e., not $\bot$) it is always the \emph{same}
symbol $\sigma\in\{0,1\}$, and moreover this $\sigma$ occurs for at least one
completion. Otherwise (either $f_{i,Q}$ produces two different non-$\bot$ outputs,
or it is always $\bot$) we declare $a$ to be \emph{bad}.

Given a local view $\ell\in \{0,1\}^{Q}$, let $a=\ell|_{L(Q)}$. We define
$g_{i,Q}(\ell)$ by:
\begin{itemize}
    \item if $a$ is \emph{good}, output the uniquely determined symbol $\sigma$
    associated with $a$;
    \item if $a$ is \emph{bad}, output a uniformly random bit.
\end{itemize}
Thus, conditioned on the light coordinates being uncorrupted, $g_{i,Q}$ depends
only on the local restriction to $L(Q)$ in the good case, while in the bad case
it deliberately ``gives up'' by outputting fresh randomness. This is exactly the
behavior we exploit in the analysis: good light patterns behave like an LDC-type
local rule, and bad patterns are charged to the overall decoding error.

The remaining part of the proof is to upper bound, for each fixed codeword $c$, the probability of sampling an $(i,c)$-nonsmoothable query set.

To do this, we use a similar resampling argument. We randomly resample all heavy coordinates, and observe that for any fixed $(i,c)$-nonsmoothable query set $Q$, with probability at least $2^{-|H_i\cap Q|}$ the resampling assigns the heavy positions in $H_i\cap Q$ to a configuration that makes the local rule output a non-$\bot$ symbol different from $b_i$. Such an event constitutes a soundness violation for the RLDC decoder. Averaging over the resampling and $Q$, the RLDC soundness guarantee forces the probability of sampling an $(i,c)$-nonsmoothable query set to be small.

\subsection{Outline of the Paper}
The rest of the paper is organized as follows. In Section~\ref{sec:pre}, we fix notation and recall the standard definitions of LDCs, LCCs and their relaxed variants (RLDCs/RLCCs). In Section~\ref{sec:main}, we prove our main structural theorem for perfectly-complete RLDCs (Theorem~\ref{thm:binary-main}). In Section~\ref{sec:lowerbound}, we combine our transformation with known LDC lower bounds to obtain lower bounds for nonadaptive perfectly-complete RLDCs in the small-soundness-error regime.

%% file: Preliminary.tex
\section{Preliminaries}\label{sec:pre}
In this section we fix notation and recall the standard local decoding notions
that will be used throughout the paper.  

\subsection{Basic notation}

We write $\NN$ for the set of positive integers. For $n \in \NN$, we denote $[n] := \{1,2,\dots,n\}$. For $t \in \NN$, we write $\binom{[n]}{t}$ for the collection of all subsets of $[n]$ of size exactly $t$. We use standard asymptotic notation $O(\cdot)$, $\Omega(\cdot)$,
$\Theta(\cdot)$, $o(\cdot)$, and $\omega(\cdot)$ with respect to the
limit $n \to \infty$.
We also use the variants $\widetilde{O}(\cdot)$,
$\widetilde{\Omega}(\cdot)$, and $\widetilde{\Theta}(\cdot)$ to denote
bounds that hold up to polylogarithmic factors in $n$.

Throughout, $\Sigma$ denotes a finite alphabet.
In most of the paper we will be interested in the binary alphabet
$\Sigma = \{0,1\}$. For $n \in \NN$, we write $\Sigma^n$ for the set of strings of length
$n$ over~$\Sigma$, and we write $\Sigma^\ast := \bigcup_{n \ge 0} \Sigma^n$
for the set of all finite strings over~$\Sigma$.
If $x \in \Sigma^n$ is a string of length $n$ and $i \in [n]$ is an
index, we denote the $i$-th symbol of $x$ by $x_i$.

For a subset $S \subseteq [n]$ and a string $x \in \Sigma^n$, we write $x|_S \in \Sigma^{|S|}$
for the restriction of $x$ to the coordinates in $S$, ordered in
increasing order.
If $X \subseteq \Sigma^n$ is a set of strings, we similarly define $X|_S := \{ x|_S : x \in X \}$. For $x,y \in \Sigma^n$, the \emph{Hamming distance} between
$x$ and $y$ is
\[
    \Delta(x,y)
    \;:=\;
    \bigl|\{\, i \in [n] : x_i \neq y_i \,\}\bigr|.
\]

A code of block length $n$ and message length $k$ over $\Sigma$ is a mapping $C : \Sigma^k \to \Sigma^n$. For $b \in \Sigma^k$ we write $c = C(b)$ for the corresponding codeword. When the message and codeword alphabets differ, we write $C:\Lambda^k\to\Sigma^n$; a local message decoder then outputs symbols in $\Lambda$. We use the definitions below with this convention when needed.
The Hamming distance of $C$ is
\[
    d := \min_{b \neq b'} \Delta\bigl(C(b),C(b')\bigr).
\]
We write $\delta_{\mathrm{dist}}(C) := d/n$ for its relative distance. For $C:\Lambda^k\to\Sigma^n$ with $|\Lambda|,|\Sigma|\ge2$, the rate is $R:=k\log|\Lambda|/(n\log|\Sigma|)$; in particular, it is $k/n$ when $\Lambda=\Sigma$.

We use $\Pr[\cdot]$ for probability and $\E[\cdot]$ for expectation with
respect to the relevant underlying randomness, which will be clear
from context. For a random variable $X$ distributed according to a distribution
$\mathcal{D}$ we write $X \sim \mathcal{D}$.


\subsection{Locally decodable/correctable codes and their relaxed variants}

We now formally define locally decodable and locally correctable codes, as well as their relaxed variants.

A $q$-query \emph{local decoder} for message symbols is a randomized oracle algorithm
\[
    \Dec : \Sigma^n \times [k] \to \Sigma
\]
such that, on input an oracle word $y \in \Sigma^n$ and an index $i \in [k]$,
the algorithm $\Dec^y(i)$ queries at most $q$ positions of $y$ and
outputs a symbol in $\Sigma$.

Similarly, a $q$-query \emph{local decoder for codeword symbols} is a randomized
oracle algorithm $\Dec : \Sigma^n \times [n] \to \Sigma$ with the analogous
semantics, where the second argument now ranges over codeword coordinates $u \in [n]$.

When the choice of each query position may depend on the previously observed
answers, we refer to $\Dec$ as \emph{adaptive}; otherwise we call it
\emph{nonadaptive}.

In the definitions below, $q\in\NN$, $\delta,s\in[0,1]$, and, for the relaxed variants, $c\in[0,1]$.

\begin{definition}[Locally decodable codes (LDCs)]
A code $C : \Sigma^k \to \Sigma^n$ is a $(q,\delta,s)$-\emph{locally
decodable code} if there exists a $q$-query local decoder
$\Dec : \Sigma^n \times [k] \to \Sigma$ such that for every
message $b \in \Sigma^k$, every index $i \in [k]$, and every received
word $y \in \Sigma^n$ with $\Delta(y,C(b)) \le \delta n$,
\[
    \Pr\bigl[\Dec^y(i) \neq b_i\bigr] \;\le\; s.
\]
\end{definition}

\begin{definition}[Locally correctable codes (LCCs)]
A code $C : \Sigma^k \to \Sigma^n$ is a $(q,\delta,s)$-\emph{locally
correctable code} if there exists a $q$-query local decoder
$\Dec : \Sigma^n \times [n] \to \Sigma$ such that for every
message $b \in \Sigma^k$, every coordinate $u \in [n]$, and every
received word $y \in \Sigma^n$ with $\Delta(y,C(b)) \le \delta n$,
\[
    \Pr\bigl[\Dec^y(u) \neq C(b)_u\bigr] \;\le\; s.
\]
\end{definition}

Relaxed variants allow the local procedure to output a special symbol
$\bot$ when the received word is corrupted, as long as it is unlikely to
output an incorrect value when the word is close to a codeword.
In the relaxed setting the decoder has output alphabet $\Sigma \cup \{\bot\}$.
On the other hand, we do not want the decoder to always output $\bot$
even when the input is an uncorrupted codeword, as this would render the
definition meaningless.
We therefore add a completeness condition requiring that, on a valid
codeword, the decoder outputs the correct symbol (rather than $\bot$)
with sufficiently large probability.

\begin{definition}[Relaxed LDCs (RLDCs)]
A code $C : \Sigma^k \to \Sigma^n$ is a $(q,\delta,c,s)$-\emph{relaxed
locally decodable code} if there exists a $q$-query local decoder $\Dec : \Sigma^n \times [k] \to \Sigma \cup \{\bot\}$
such that for every $b \in \Sigma^k$, every $i \in [k]$, and every
$y \in \Sigma^n$ with $\Delta(y,C(b)) \le \delta n$:
\begin{enumerate}
    \item (Completeness) For the uncorrupted codeword,
    \[
        \Pr\bigl[\Dec^{C(b)}(i) = b_i\bigr] \;\ge\; c.
    \]
    \item (Soundness) On corrupted words,
    \[
        \Pr\bigl[\Dec^y(i) \notin \{b_i,\bot\}\bigr] \;\le\; s.
    \]
\end{enumerate}
\end{definition}

\begin{definition}[Relaxed LCCs (RLCCs)]
A code $C : \Sigma^k \to \Sigma^n$ is a $(q,\delta,c,s)$-\emph{relaxed
locally correctable code} if there exists a $q$-query local decoder $\Dec : \Sigma^n \times [n] \to \Sigma \cup \{\bot\}$
such that for every $b \in \Sigma^k$, every $u \in [n]$, and every
$y \in \Sigma^n$ with $\Delta(y,C(b)) \le \delta n$:
\begin{enumerate}
    \item (Completeness) For the uncorrupted codeword,
    \[
        \Pr\bigl[\Dec^{C(b)}(u) = C(b)_u\bigr] \;\ge\; c.
    \]
    \item (Soundness) On corrupted words,
    \[
        \Pr\bigl[\Dec^y(u) \notin \{C(b)_u,\bot\}\bigr] \;\le\; s.
    \]
\end{enumerate}
\end{definition}

When $c = 1$ we say that the relaxed decoder has \emph{perfect completeness}.
Sometimes the term ``completeness'' is also used for LDCs/LCCs themselves, but in general we do not impose any
explicit completeness requirement in the definition of an LDC/LCC, and in this
paper our LDC/LCC definitions are purely in terms of the error parameter $s$.

When the relaxed decoder has perfect completeness ($c = 1$), one can derive a few
useful structural properties that we will use without proof; see, for example,
\cite{grigorescu2025relaxed} for detailed arguments.

First, an RLDC/RLCC with perfect completeness can be assumed to have a \emph{canonical} behavior. Informally, for an index $i \in [k]$ (or a coordinate $u \in [n]$ in the RLCC case), once the decoder has fixed its query set $Q$ and observed the local view $y|_Q$, it proceeds as follows. Perfect completeness ensures that the corresponding symbol is unique whenever such a codeword exists.
\begin{enumerate}
    \item If there exists a codeword $c = C(b)$ such that $c|_Q = y|_Q$, output the unique corresponding symbol $b_i$ (or $c_u$ in the RLCC setting).
    \item If no such codeword exists, output $\bot$.
\end{enumerate}

\begin{proposition}[Canonical form under perfect completeness]
\label{prop:canonical}
Let $C : \Sigma^k \to \Sigma^n$ be a $(q,\delta,1,s)$-RLDC with decoder
$\Dec_1 : \Sigma^n \times [k] \to \Sigma \cup \{\bot\}$.
Then there exists a decoder $\Dec_2$ such that $C$ is still a
$(q,\delta,1,s)$-RLDC with respect to $\Dec_2$, and $\Dec_2$ is canonical in the
sense described above.
Moreover, if $\Dec_1$ is nonadaptive, then $\Dec_2$ can also be chosen
nonadaptive.
An analogous statement holds for RLCCs, with $i \in [k]$ replaced by
$u \in [n]$ and $b_i$ replaced by $C(b)_u$.
\end{proposition}

Starting from any perfect-completeness RLDC/RLCC, one can thus wrap the original
decoder by this canonical post-processing.
On true codewords, the output is unchanged, so completeness remains $1$, and
replacing some noncanonical outputs by $\bot$ can only improve (or leave
unchanged) the relaxed soundness parameter.

There are standard soundness amplification techniques for LDCs and RLDCs
based on independent repetition of the decoder.

\begin{proposition}[Soundness amplification for LDCs]
Let $C : \Sigma^k \to \Sigma^n$ be a $(q,\delta,s)$-LDC with $0\le s\le 1/2$. Then, for every integer $t \ge 1$, $C$ is also a $\bigl(tq,\delta,\bigl(2\sqrt{s(1-s)}\bigr)^t\bigr)$-LDC. The same statement holds for LCCs.
\end{proposition}

In the relaxed setting, perfect completeness allows an even simpler
amplification of the soundness by independent repetition.

\begin{proposition}[Soundness amplification for RLDCs {\cite[Observation 1.9]{grigorescu2025relaxed}}]
\label{prop:amplification}
Let $C : \Sigma^k \to \Sigma^n$ be a $(q,\delta,1,s)$-RLDC.
Then, for every integer $t \ge 1$, $C$ is also a $(tq,\delta,1,s^t)$-RLDC.
The same statement holds for RLCCs.
\end{proposition}
One convenient instantiation of Proposition~\ref{prop:amplification} is the
following: on input $(y,i)$, define a new decoder that runs $\Dec^y(i)$
independently $t$ times, and outputs a symbol $\sigma \in \Sigma$ only if all
$t$ invocations output $\sigma$; otherwise it outputs $\bot$. The same amplification construction applies verbatim to RLCCs.

%% file: Results.tex
\section{From RLDCs to LDCs in the Small-Soundness-Error Regime}
\label{sec:main}

In this section we prove our main result, Theorem~\ref{thm:binary-main}, which shows that a nonadaptive RLDC with perfect completeness and sufficiently small soundness error yields an LDC with comparable parameters. We focus on the RLDC/LDC setting; the RLCC/LCC analogue follows by the same argument after replacing the decoded message symbol $b_i$ by the target codeword symbol $C(b)_u$. We prove the result over an arbitrary finite alphabet $\Sigma$; the binary statement is obtained by setting $\Sigma=\{0,1\}$.

\begin{theorem}[Theorem \ref{thm:binary-main} for any alphabet]\label{thm:derived_ldc}
Let $\Sigma$ be a finite alphabet with $|\Sigma|\ge2$, let $q\in\mathbb N$, and let $\delta\in(0,1]$. Suppose $C : \Sigma^{k} \to \Sigma^{n}$ is a $(q,\delta,1,s)$-RLDC with a nonadaptive decoder, where $s < |\Sigma|^{-q}$. Then, for any error radius $0<r< \frac{\delta(|\Sigma|-1)}{q|\Sigma|} \bigl(1-s|\Sigma|^q\bigr)$, the code $C$ is also a
\[
    \bigl(q, r,\, \varepsilon\bigr)\text{-LDC}
\]
for $\varepsilon = s |\Sigma|^{q}\Bigl(1 - \frac{1}{|\Sigma|}\Bigr) + \frac{rq}{\delta}$.
\end{theorem}

\begin{corollary}[Mixed-alphabet form]
\label{cor:mixed-alphabet-derived-ldc}
Let $\Lambda$ and $\Sigma$ be finite alphabets with $|\Lambda|,|\Sigma|\ge2$, let $q\in\mathbb N$ and $\delta\in(0,1]$, and let $C:\Lambda^k\to\Sigma^n$ be a nonadaptive $(q,\delta,1,s)$-RLDC, where $s<|\Sigma|^{-q}$. Here the decoder has type
\[
    \Dec:\Sigma^n\times[k]\to\Lambda\cup\{\bot\}.
\]
Then, for every
\[
    0<r<
    \frac{\delta}{q}
    \left(1-\frac1{|\Lambda|}\right)
    \left(1-s|\Sigma|^q\right),
\]
the code is a $(q,r,\varepsilon)$-LDC, where
\[
    \varepsilon
    =
    s|\Sigma|^q\left(1-\frac1{|\Lambda|}\right)
    +\frac{rq}{\delta}.
\]
\end{corollary}

\begin{proof}
Canonicalization is unchanged: a local view consistent with codewords determines a unique symbol of $\Lambda$. Heavy-coordinate resampling still uses the codeword alphabet $\Sigma$, while a bad light pattern is answered uniformly from $\Lambda$. Applying the construction in the proof below gives a decoder $\Dec'$ satisfying
\[
\Pr[\Dec'^y(i)=b_i]
\ge \frac1{|\Lambda|}
   +\left(1-\frac1{|\Lambda|}\right)(1-s|\Sigma|^q)
   -\frac{rq}{\delta}
=1-\varepsilon,
\]
which proves the claim.
\end{proof}

Henceforth, we assume that the RLDC decoder is nonadaptive. This assumption is part of the theorem statement. Since the decoder has perfect completeness, we may also pass to canonical form using Proposition~\ref{prop:canonical}, which preserves nonadaptivity and the parameters.

\begin{remark}
The nonadaptivity assumption in Theorem~\ref{thm:derived_ldc} is essential for our proof. The argument uses a fixed distribution $\sQ_i$ over query sets and a fixed truth table $f_{i,Q}$ for each sampled query set.
\end{remark}

By Proposition~\ref{prop:canonical}, because the decoder has perfect completeness, we may assume without loss of generality that this nonadaptive decoder is canonical.

\paragraph{Characterization of RLDC decoders}
Before proving our main theorem, we introduce some additional notation
and give a more precise characterization of RLDC decoders.

Let $C : \Sigma^k \to \Sigma^n$ be a code over a finite alphabet
$\Sigma$, and fix an index $i \in [k]$.
Since the decoder is assumed nonadaptive and has perfect completeness, we may pass to canonical form using Proposition~\ref{prop:canonical}. An RLDC decoder for the $i$-th message symbol is completely specified by:
\begin{itemize}
    \item a distribution $\sQ_i$ over subsets $Q \subseteq [n]$ of size at most $q$, and
    \item for each $Q$ in the support of $\sQ_i$, a local rule
          $f_{i,Q} : \Sigma^Q \to \Sigma \cup \{\bot\}$.
\end{itemize}
We will refer to the function $f_{i,Q}$ as the \emph{truth table} of the
decoder for index $i$ on the query set $Q$: for every local view
$z \in \Sigma^Q$, the value $f_{i,Q}(z)$ is exactly the output that the
canonical decoder would produce when its sampled query set is $Q$ and the observed local view is $z$.

Given a received word $y \in \Sigma^n$, the decoder samples
$Q \sim \sQ_i$ and outputs $f_{i,Q}(y|_Q)$.
In particular, the behavior of the decoder on index $i$ is completely
captured by the query distribution $\sQ_i$ together with the family of
truth tables $\{f_{i,Q}\}$.
We denote the support of $\sQ_i$ by $\supp(\sQ_i) \subseteq \{Q\subseteq[n]: |Q|\le q\}$. For each coordinate $j \in [n]$, we let
\[
    p_i(j) := \Pr_{Q \sim \sQ_i}[\, j \in Q \,]
\]
denote the probability that $j$ is queried when decoding index $i$.

\paragraph{\textbf{Heavy and light coordinates}}
For each $i \in [k]$, we partition the coordinates in $[n]$ into \emph{heavy}
and \emph{light} positions. We define
\[
    H_i := \{\, j \in [n] : p_i(j) > q / (\delta n) \,\},
    \qquad
    L_i := [n] \setminus H_i.
\]
For a query set $Q \subseteq [n]$, we write
\[
    H_i(Q) := Q \cap H_i,
    \qquad
    L_i(Q) := Q \cap L_i.
\]
When the context is clear and we are decoding $b_i$, we may omit the subscript $i$ in the notation and simply write $\sQ$, $p(j)$, $H$, $L$, $H(Q)$, and $L(Q)$.
By the definition of heavy coordinates, we can easily bound the size of
$H$ by $|H| \le \delta n$ since
\[
    |H| \cdot \frac{q}{\delta n}
    \;\le\;
    \sum_{j \in H} p(j)
    \;\le\;
    \sum_{j=1}^n p(j)
    \;\le\;
    q.
\]

Next we introduce the notion of smooth queries.
Our definition differs from that of \cite{grigorescu2025relaxed}:
there, the distribution $\sQ_i$ is directly partitioned into a
codeword-independent smoothable part and nonsmoothable part, whereas in
our setting the smoothability of a query set $Q$ depends not only on the
decoded index $i$ but also on the particular codeword $c$.
Later we will explain how, despite this dependence on $c$, one can
decide whether $Q$ is smoothable using only the local view on the light
coordinates and the truth table, without reading the entire received
word.

\begin{definition}[$(i,c)$-smoothable query sets]
Fix a message $b \in \Sigma^k$ and its codeword $c = C(b) \in \Sigma^n$.
We say that a query set $Q$ is \emph{$(i,c)$-smoothable} if the light
restriction $c|_{L(Q)}$ uniquely determines the symbol $b_i$ among all
codewords, that is,
\[
    \forall b' \in \Sigma^k,\qquad
    C(b')|_{L(Q)} = c|_{L(Q)}
    \;\Longrightarrow\;
    b'_i = b_i.
\]
Equivalently, $Q$ is \emph{$(i,c)$-nonsmoothable} if there exists a
message $b' \in \Sigma^k$ with codeword $c' = C(b')$ such that
\[
    c'|_{L(Q)} = c|_{L(Q)}
    \quad\text{and}\quad
    b'_i \neq b_i.
\]
\end{definition}
Let $\NS_{i,c}\subseteq\supp(\sQ_i)$ denote the family of $(i,c)$-nonsmoothable query sets.

Next, we illustrate the definition of $(i,c)$-smoothability on a simple
binary example, and explain why the dependence on the codeword $c$ is
essential for handling nonlinear codes.

\begin{example}\label{ex:truthtable}
For this example we take $\Sigma = \{0,1\}$. Fix an index $i \in [k]$ and a query set
\[
    Q = \{j_1,j_2,j_3\},
    \qquad\text{with}\qquad
    j_1 \in H_i,\;\; j_2,j_3 \in L_i.
\]
Consider the local rule $f_{i,Q} : \{0,1\}^3 \to \{0,1,\bot\}$ given by the following truth table:
\[
\begin{array}{c|c}
z_1z_2z_3 & f_{i,Q}(z_1,z_2,z_3)\\\hline
000 & 1\\
001 & 0\\
010 & 1\\
011 & \bot\\
100 & 0\\
101 & 0\\
110 & \bot\\
111 & \bot
\end{array}
\]

Now fix a codeword $c = C(b)$ and examine the four possible patterns of the light restriction $c|_{\{j_2,j_3\}}$.
\begin{enumerate}
    \item $c|_{\{j_2,j_3\}} = 00$. Since $f_{i,Q}(0,0,0) = 1$ but $f_{i,Q}(1,0,0) = 0$, the decoder's output depends on the heavy coordinate $j_1$ even after fixing the light coordinates. Thus $Q$ is \emph{$(i,c)$-nonsmoothable} when $c|_{\{j_2,j_3\}} = 00$.
    \item $c|_{\{j_2,j_3\}} = 01$. Here $f_{i,Q}(0,0,1) = f_{i,Q}(1,0,1) = 0$, so once the light pattern is $01$ the output is always $0$, independent of the heavy coordinate $j_1$. Hence $Q$ is \emph{$(i,c)$-smoothable} when $c|_{\{j_2,j_3\}} = 01$.
    \item $c|_{\{j_2,j_3\}} = 10$. We have $f_{i,Q}(0,1,0) = 1$ while $f_{i,Q}(1,1,0) = \bot$. By perfect completeness, on a correct codeword $c'$ the decoder never outputs $\bot$, so the local view $(1,1,0)$ cannot arise from any uncorrupted codeword. In particular, the only consistent output on a correct codeword with light pattern $10$ is $1$, and $Q$ is again \emph{$(i,c)$-smoothable} when $c|_{\{j_2,j_3\}} = 10$.
    \item $c|_{\{j_2,j_3\}} = 11$. In this case both rows with light pattern $11$ (`$011$' and `$111$') map to $\bot$. Under perfect completeness, this means that no uncorrupted codeword $c'$ can satisfy $c'|_{\{j_2,j_3\}} = 11$.
\end{enumerate}

Thus $Q$ is smoothable for light patterns $01$ and $10$, nonsmoothable for $00$, and inconsistent with every codeword for $11$. The same $Q$ can therefore behave differently for different codewords, so smoothability must depend on $c$. If the light coordinates are uncorrupted, this dependence is still local. For a fixed light restriction $c|_{L(Q)}$, $Q$ is $(i,c)$-nonsmoothable if and only if there exist two local views $\ell,\ell' \in \{0,1\}^Q$ with
\[
    \ell|_{L(Q)} = \ell'|_{L(Q)} = c|_{L(Q)}
    \quad\text{and}\quad
    f_{i,Q}(\ell), f_{i,Q}(\ell') \in \{0,1\}
    \text{ with } f_{i,Q}(\ell) \neq f_{i,Q}(\ell').
\]
Hence the smoothability test uses only the observed light symbols and the truth table $f_{i,Q}$.
\end{example}

We now convert the RLDC decoder into an LDC decoder by making the local rule depend only on the light pattern whenever that pattern determines a unique output, and by replacing the remaining cases by a random symbol.
Fix $i \in [k]$ and a query set $Q$ in the support of $\sQ_i$.
We define a randomized local rule
\[
    g_{i,Q} : \Sigma^Q \to \Sigma.
\]

For a pattern on the light coordinates $a \in \Sigma^{L(Q)}$, consider
all completions $z \in \Sigma^Q$ with $z|_{L(Q)} = a$.

We call $a$ \emph{bad} if either
\begin{itemize}
    \item there exist $z,z' \in \Sigma^Q$ with $z|_{L(Q)} = z'|_{L(Q)} = a$ and
    \[
        f_{i,Q}(z), f_{i,Q}(z') \in \Sigma
        \quad\text{and}\quad
        f_{i,Q}(z) \neq f_{i,Q}(z'),
    \]
    i.e., the same light pattern $a$ leads to two different non-$\bot$ outputs of $f_{i,Q}$; or
    \item for every $z \in \Sigma^Q$ with $z|_{L(Q)} = a$, we have
    \[
        f_{i,Q}(z) = \bot,
    \]
    i.e., the light pattern $a$ never yields a non-$\bot$ output.
\end{itemize}

We call $a$ \emph{good} if it is not bad.
Equivalently, $a$ is good if there exists a \emph{unique} symbol
$\sigma \in \Sigma$ such that
\[
    \exists z \in \Sigma^Q \text{ with } z|_{L(Q)} = a
    \text{ and } f_{i,Q}(z) = \sigma,
\]
and for every $z' \in \Sigma^Q$ with $z'|_{L(Q)} = a$ we have $f_{i,Q}(z') \in \{\sigma, \bot\}$.

We now define $g_{i,Q}$.
Given a local view $\ell \in \Sigma^Q$, let $a := \ell|_{L(Q)}$ be its
restriction to the light coordinates.

\begin{itemize}
    \item If $a$ is bad, we define $g_{i,Q}(\ell)$ to be a uniformly random
          symbol in $\Sigma$.
    \item If $a$ is good, let $\sigma \in \Sigma$ be the unique symbol
          guaranteed by the definition of a good pattern $a$, namely the
          unique $\sigma$ for which there exists $z \in \Sigma^Q$ with
          $z|_{L(Q)} = a$ and $f_{i,Q}(z) = \sigma$.
          In this case we set
          \[
              g_{i,Q}(\ell) := \sigma.
          \]
\end{itemize}
For a light pattern induced by an actual codeword $c$, the all-$\bot$ case cannot occur by canonicality and perfect completeness. Hence, $Q$ is $(i,c)$-nonsmoothable if and only if the light pattern $c|_{L(Q)}$ is bad. The derived family $\{g_{i,Q}\}$ will serve as the local rules of our LDC decoder.
\begin{proof}[Proof of Theorem~\ref{thm:derived_ldc}]
Define a decoder $\Dec'$ for message symbols as follows. Given oracle access to a received word $y \in \Sigma^n$ and an index $i \in [k]$:
\begin{enumerate}
    \item Sample a query set $Q \sim \sQ_i$.
    \item Output $g_{i,Q}(y|_Q)$.
\end{enumerate}
$\Dec'$ makes at most $q$ queries. Fix a codeword $c = C(b)$, an index $i\in[k]$, and a received word $y$ with $\Delta(y,c) \le r n$. We analyze the success probability of $\Dec'$ on input $(y,i)$. Let
\[
    \Err := \{ j \in [n] : y_j \neq c_j \}
\]
denote the set of corrupted coordinates. We distinguish three cases:

\begin{enumerate}
    \item \textbf{Light corruption event:}
          $L(Q) \cap \Err \neq \emptyset$.
          In this case at least one light coordinate is corrupted.
          We denote
          \[
              \eta := \Pr_{Q \sim \sQ_i}\bigl[\, L(Q) \cap \Err \neq \emptyset \,\bigr].
          \]

    \item \textbf{Light clean but nonsmoothable:}
          $y|_{L(Q)} = c|_{L(Q)}$ and $Q \in \NS_{i,c}$.
          In this case all light coordinates are correct, but the query
          set is $(i,c)$-nonsmoothable.
          By construction of $g_{i,Q}$, when the light pattern
          corresponds to a nonsmoothable query we output a uniform
          symbol in $\Sigma$, so $\Dec'$ succeeds with probability
          exactly $1/|\Sigma|$.

    \item \textbf{Light clean and smoothable:}
          $y|_{L(Q)} = c|_{L(Q)}$ and $Q$ is $(i,c)$-smoothable.
          In this case the light coordinates are uncorrupted and the
          light pattern uniquely determines $b_i$.
          By definition of $g_{i,Q}$, we have
          \[
              g_{i,Q}(y|_Q) = g_{i,Q}(c|_Q) = b_i,
          \]
          so the decoder succeeds with probability $1$.
\end{enumerate}
Let
\[
    E := \{\,L(Q)\cap \Err=\emptyset\,\},
    \qquad
    N := \{\,Q\in\NS_{i,c}\,\},
\]
and recall
\[
    \eta := \Pr_{Q \sim \sQ_i}\bigl[\, L(Q) \cap \Err \neq \emptyset \,\bigr].
\]
On $E\cap \overline{N}$ the decoder succeeds with probability $1$, and on $E\cap N$ it succeeds with probability $1/|\Sigma|$. Therefore
\[
\begin{aligned}
\Pr\bigl[\Dec'^y(i) = b_i\bigr]
&\ge
   \Pr[E\cap \overline{N}]
   + \frac{1}{|\Sigma|}\Pr[E\cap N].
\end{aligned}
\]
Since
\[
    \Pr[E\cap \overline{N}]
    =
    \Pr[E]-\Pr[E\cap N]
    =
    1-\eta-\Pr[E\cap N],
\]
the last display gives
\[
\begin{aligned}
\Pr\bigl[\Dec'^y(i) = b_i\bigr]
&\ge
   1-\eta-\Pr[E\cap N]
   + \frac{1}{|\Sigma|}\Pr[E\cap N] \\
&=
   1-\eta-\Bigl(1-\frac{1}{|\Sigma|}\Bigr)\Pr[E\cap N] \\
&\ge
   1-\eta-\Bigl(1-\frac{1}{|\Sigma|}\Bigr)\Pr[N].
\end{aligned}
\]
Equivalently, writing
\[
    \alpha := \Pr_{Q \sim \sQ_i}\bigl[\, Q \text{ is $(i,c)$-smoothable} \,\bigr]
            = 1 - \Pr[N],
\]
we have
\[
\begin{aligned}
\Pr\bigl[\Dec'^y(i) = b_i\bigr]
&\ge
   1-\eta-\Bigl(1-\frac{1}{|\Sigma|}\Bigr)(1-\alpha) \\
&=
   \frac{1}{|\Sigma|}
   + \Bigl(1-\frac{1}{|\Sigma|}\Bigr)\alpha - \eta.
\end{aligned}
\]
We next bound $\alpha$ and $\eta$.

\begin{lemma}
\label{lem:ns_mass_bound-again}
For any index $i \in [k]$ and every codeword $c = C(b)$,
\[
    \Pr_{Q \sim \sQ_i}\bigl[\, Q \in \NS_{i,c} \,\bigr]
    \;\le\;
    s \cdot |\Sigma|^{q}.
\]
\end{lemma}

\begin{lemma}\label{lem:corrupted-light}
For any index $i\in[k]$ and any corruption set $\Err\subseteq[n]$ with $|\Err|\leq rn$, if $\eta:=\Pr_{Q\sim\sQ_i}[\,L(Q)\cap \Err\neq\emptyset\,]$, then
\[
    \eta \leq \frac{rq}{\delta}.
\]
\end{lemma}
We postpone the proofs of Lemmas~\ref{lem:ns_mass_bound-again}
and~\ref{lem:corrupted-light}.
Assuming them, we have
\[
    \alpha = 1 - \Pr[Q \in \NS_{i,c}]
    \;\ge\;
    1 - s |\Sigma|^{q},
\]
and
\[
    \eta \le \frac{rq}{\delta}.
\]
Plugging these bounds into the previous inequality yields
\[
\begin{aligned}
\Pr\bigl[\Dec'^y(i) = b_i\bigr]
&\ge
   \frac{1}{|\Sigma|}
   + \Bigl(1-\frac{1}{|\Sigma|}\Bigr)\bigl(1 - s |\Sigma|^{q}\bigr)
   - \frac{rq}{\delta} \\
&=
   1 - s |\Sigma|^{q}\Bigl(1 - \frac{1}{|\Sigma|}\Bigr)
   - \frac{rq}{\delta}.
\end{aligned}
\]
This proves Theorem~\ref{thm:derived_ldc}. In particular, for $\Sigma = \{0,1\}$ this simplifies to
\[
    \Pr\bigl[\Dec'^y(i) = b_i\bigr]
    \;\ge\;
    1 - \frac{s 2^{q}}{2} - \frac{rq}{\delta},
\]
as claimed in Theorem~\ref{thm:binary-main}.
\end{proof}

\subsection{Proof of Lemma~\ref{lem:ns_mass_bound-again}: Bounding the probability of nonsmoothable queries}

In this subsection we show that, for any fixed codeword $c$, when we sample
a query set $Q \sim \sQ_i$, the event that $Q$ is $(i,c)$-nonsmoothable
can occur only with small probability.

\begin{proof}[Proof of Lemma \ref{lem:ns_mass_bound-again}]
Fix $i \in [k]$ and a codeword $c = C(b)$.
Define a random received word $\widetilde c \in \Sigma^n$ by resampling the heavy
coordinates:
\[
    \widetilde c_j :=
    \begin{cases}
        c_j, & j \in L_i,\\
        \text{an independent uniform symbol in }\Sigma, & j \in H_i.
    \end{cases}
\]
By construction, the only positions where $\widetilde c$ may differ from $c$ are
those in $H_i$, so
\[
    \Delta(\widetilde c,c) \le |H_i| \le \delta n,
\]
and therefore $\widetilde c$ is always within the RLDC noise tolerance from $c$.

We first show that each nonsmoothable query set $Q$ is ``fooled'' with
noticeable probability.

\begin{lemma}
\label{lem:nonsmoothable_fool_Sigma}
Fix $i$ and a codeword $c = C(b)$. For every $(i,c)$-nonsmoothable query set $Q\in\supp(\sQ_i)$,
\[
    \Pr_{\widetilde c}\bigl[\, f_{i,Q}(\widetilde c|_Q) \notin \{b_i,\bot\} \,\bigr]
    \;\ge\;
    |\Sigma|^{-|H(Q)|}
    \;\ge\;
    |\Sigma|^{-q}.
\]
\end{lemma}

\begin{proof}
Since $Q$ is $(i,c)$-nonsmoothable, there exists a message
$b'' \in \Sigma^k$ with codeword $c'' = C(b'')$ such that
\[
    c''|_{L(Q)} = c|_{L(Q)}
    \quad\text{and}\quad
    b''_i \neq b_i.
\]
By the definition of $\widetilde c$, the restriction $\widetilde c|_{H(Q)}$ is uniform over
$\Sigma^{H(Q)}$ and independent of $c|_{L(Q)}$.
Hence
\[
    \Pr\bigl[\, \widetilde c|_{H(Q)} = c''|_{H(Q)} \,\bigr]
    = |\Sigma|^{-|H(Q)|}.
\]
On this event we have $\widetilde c|_Q = c''|_Q$.
By perfect completeness, whenever the local view is consistent with a
codeword $c''$, the canonical decoder outputs $b''_i$ (and never $\bot$),
so
\[
    f_{i,Q}(\widetilde c|_Q) = f_{i,Q}(c''|_Q) = b''_i \neq b_i.
\]
Thus
\[
    \Pr_{\widetilde c}\bigl[\, f_{i,Q}(\widetilde c|_Q) \notin \{b_i,\bot\} \,\bigr]
    \;\ge\;
    \Pr\bigl[\, \widetilde c|_{H(Q)} = c''|_{H(Q)} \,\bigr]
    = |\Sigma|^{-|H(Q)|}
    \;\ge\;
    |\Sigma|^{-q},
\]
as claimed.
\end{proof}

Recall that $\NS_{i,c}$ is the family of $(i,c)$-nonsmoothable query sets.
Sampling $Q \sim \sQ_i$ and $\widetilde c$ as above, we obtain
\[
    \Pr_{Q \sim \sQ_i,\, \widetilde c}\bigl[\, f_{i,Q}(\widetilde c|_Q) \notin \{b_i,\bot\} \,\bigr]
    \;\ge\;
    \Pr_{Q \sim \sQ_i}[\, Q \in \NS_{i,c} \,] \cdot |\Sigma|^{-q},
\]
where we used Lemma~\ref{lem:nonsmoothable_fool_Sigma} and the fact that
$|\Sigma|^{-|H(Q)|} \ge |\Sigma|^{-q}$.

On the other hand, since $\Delta(\widetilde c,c) \le \delta n$, the relaxed soundness
condition for the $(q,\delta,1,s)$-RLDC implies
\[
    \Pr_{Q \sim \sQ_i,\, \widetilde c}\bigl[\, f_{i,Q}(\widetilde c|_Q) \notin \{b_i,\bot\} \,\bigr]
    \;=\;
    \E_{\widetilde c}\!\left[
        \Pr\bigl[\, \Dec^{\widetilde c}(i) \notin \{b_i,\bot\} \,\bigm|\, \widetilde c\bigr]
    \right]
    \;\le\;
    s.
\]
Combining the two bounds yields
\[
    \Pr_{Q \sim \sQ_i}[\, Q \in \NS_{i,c} \,]
    \;\le\;
    s \cdot |\Sigma|^{q},
\]
\end{proof}

\subsection{Proof of Lemma~\ref{lem:corrupted-light}: Bounding the probability of light error}
In this subsection we show that, when the overall error rate is small,
a random query set is unlikely to contain a corrupted light coordinate.
\begin{proof}[Proof of Lemma~\ref{lem:corrupted-light}]
Recall that
\[
    L(Q) := Q \cap L_i
\]
and
\[
    p_i(j) := \Pr_{Q \sim \sQ_i}\bigl[j \in Q\bigr],
\]
and that every light coordinate satisfies
\[
    p_i(j) \le \frac{q}{\delta n}
    \qquad\text{for all } j \in L_i.
\]

Let $\Err \subseteq [n]$ be the corruption set with $|\Err| \le rn$, and define
\[
    X := |L(Q) \cap \Err|.
\]
Using linearity of expectation, we have
\[
\E[X] = \E\Bigl[\,\sum_{j \in \Err \cap L_i} \mathbf{1}\{j \in Q\}\Bigr] = \sum_{j \in \Err \cap L_i} \Pr[j \in Q] = \sum_{j \in \Err \cap L_i} p_i(j).
\]
Since $\Err \cap L_i \subseteq L_i$ and each $j \in L_i$ satisfies
$p_i(j) \le \frac{q}{\delta n}$, it follows that
\[
\E[X] \le |\Err \cap L_i| \cdot \frac{q}{\delta n} \le |\Err| \cdot \frac{q}{\delta n} \le rn \cdot \frac{q}{\delta n}
= \frac{rq}{\delta}.
\]
By Markov's inequality,
\[
    \eta
    \;=\;
    \Pr[X \ge 1]
    \;\le\;
    \E[X]
    \;\le\;
    \frac{rq}{\delta},
\]
as claimed.
\end{proof}

%% file: LBPCPPRLDC.tex
\section{Lower Bounds for RLDCs}\label{sec:lowerbound}
Combining our main theorem with the odd-query LDC lower bound of \cite{janzer2024k} gives the following bound for nonadaptive RLDCs.
\begin{corollary}
For every constant odd integer $q\ge 3$ and constant $\delta\in(0,1)$, if $C:\{0,1\}^k\to\{0,1\}^n$ is a nonadaptive $(q,\delta,1,s)$-RLDC with $s\le 2^{-(q+1)}$, then $k\le O(n^{1-2/q}\log n)$.
\end{corollary}
\begin{proof}
Take $r=0.01\delta/q$. Since $s\le 2^{-(q+1)}$,
\[
    r<\frac{\delta}{4q}
      \le\frac{\delta}{2q}(1-s2^q),
\]
so the radius condition in Theorem~\ref{thm:derived_ldc} holds. Theorem~\ref{thm:derived_ldc} gives a $(q,r,\varepsilon)$-LDC with the same block length and message length, where
\[
    \varepsilon
    =\frac{s2^q}{2}+\frac{rq}{\delta}
    \le \frac14+0.01
    <0.3.
\]
The claim follows from the odd-query LDC lower bound of \cite{janzer2024k}, since the decoding success probability is greater than $0.7=1/2+0.2$.
\end{proof}